\documentclass[12pt]{article}
\usepackage{amsmath,amsfonts, amsthm,amssymb,hyperref,bbm,eucal,verbatim}

\newtheorem{thm}{Theorem}
\newtheorem{lemma}{Lemma}[section]

\newtheorem*{rmk}{Remark}

\newtheorem{cl}[lemma]{Claim}

\def\wig{\rho_\text{sc}}
\def\tr{\mathrm{tr}}
\def\z{{\bf z}}
\def\x{{\bf x}}
\def\y{{\bf y}}
\def\e{{\bf e}}
\def\xx{\tilde{\x}}
\def\yy{\tilde{\y}}
\def\bpsi{{\boldsymbol \psi}}
\def\bphi{{\boldsymbol \phi}}
\def\tphi{\widetilde{\boldsymbol \phi}}

\def\E{E_\epsilon}

\begin{document}

\title{Density of states for GUE, GOE, and interpolating ensembles through supersymmetric approach}
\author{Mira Shamis\footnote{Department of Mathematics,  Princeton University,  Princeton NJ 08544, USA and Institute for Advanced Study, Einstein Dr., Princeton, NJ 08540 .
e-mail: mshamis.princeton.edu. Supported in part by NSF grants PHY-1104596 and DMS-1128155.}}
\maketitle

\begin{abstract}
We use the supersymmetric formalism to derive an integral formula for the density of states of the
Gaussian Orthogonal Ensemble, and then apply saddle-point analysis to give a new derivation
of the $1/N$-correction to Wigner's law. This extends the work of Disertori on the Gaussian
Unitary Ensemble. We also apply our method to the interpolating ensembles of Mehta--Pandey.
\end{abstract}

\section{Introduction}

In this note we study the density of states for the Gaussian Unitary Ensemble
(GUE) and the Gaussian Orthogonal Ensemble (GOE). Both are ensembles of $N \times N$
Hermitian random matrices $H$, so that the joint distribution of the entries is
centered Gaussian, and the covariance of the entries is given by
\begin{equation}\label{eq:cov} \langle H_{ij} H_{kl} \rangle =
\begin{cases}
 \frac{1}{N} \delta(jk) \delta(il)~, &\text{GUE}\\
 \frac{1}{N} \left( \delta(jk) \delta(il) + \delta(ik) \delta(jl)\right)~, &\text{GOE}
\end{cases}~.
\end{equation}
Here $\langle \cdot \rangle$ denotes the average (expectation), and $\delta$
denotes the Kronecker delta,
\[ \delta(ij) = \begin{cases}
1~, &i=j\\
0~, &i \neq j \end{cases}~.\]
In particular, the GUE entries are complex (the diagonal elements are real), whereas
the GOE entries are real.

The density of states $\rho(E)$ is defined by
\[\rho(E) = \frac{d}{dE} \left\langle \frac{1}{N} \# \left\{ \text{eigenvalues $\leq E$} \right\} \right\rangle~.\]
Let
\[ \wig(E) = \frac{1}{\pi} \sqrt{(1-E^2/4)_+} \]
be the (Wigner) semicircle density. We give new proofs for the following two theorems (see below for
the history of these and related results):

\begin{thm}\label{thm:gue} For GUE,
\[ \rho(E) =
\wig(E) - \frac{(-1)^N}{4\pi^3N\wig(E)^2} \cos\left[N\left(E\sqrt{1-\frac{E^2}{4}}+2\arcsin\frac{E}{2}\right)\right]
  + O(N^{-3/2}) \]
for $|E| < 2-\delta$, and the implicit constant in the $O$-notation depends only on $\delta >0$.
\end{thm}

\begin{thm}\label{thm:goe} For GOE,
\[ \rho(E) = \wig(E) - \frac{1}{4\pi^2N\wig(E)} + O(N^{-3/2}) \]
for $|E| < 2-\delta$, and the implicit constant in the $O$-notation depends only on $\delta >0$.
\end{thm}

\begin{rmk} The oscillatory term in the expansion corresponding to GOE is of order $N^{-2}$,
see Kalisch and Braak \cite{KB}. It can also be derived by our methods.
\end{rmk}

We also consider the interpolating ensembles of Mehta and Pandey \cite{MP}, which are given by
\[ \sqrt{r} \, \mathrm{GUE} + \sqrt{1-r} \, \mathrm{GOE}~, \quad 0 \leq r \leq 1~. \]
The case $r = 0$ corresponds to the GOE, whereas $r=1$ corresponds to the GUE (in the notation of \cite{MP}, $r = \alpha^2$). More 
explicitly, 
\begin{equation}\label{eq:cov'} \langle H_{ij} H_{kl} \rangle =
 \frac{1}{N} \left( \delta(jk) \delta(il) + (1-r) \delta(ik) \delta(jl)\right)~.
\end{equation}
We prove:
\begin{thm}\label{thm:interp}
For the interpolating ensemble (\ref{eq:cov'}), 
\[\begin{split} \rho(E) &=
\wig(E) - 
\frac{(-1)^N r }{4\pi^3N\wig(E)^2} \Re \left\{ \frac{e^{-N \left(i E \sqrt{1-E^2/4} 
	+ 2 \ln(iE/2 + \sqrt{1-E^2/4})\right)}}{(1-r)(1-E^2/2 - iE \sqrt{1-E^2/4}) + 1} \right\}\\
  &\quad- \frac{1-r}{2\pi^2 N \wig(E)}\Re \Big\{ (2 \sqrt{1-E^2/4}-r(-iE/2 + \sqrt{1-E^2/4})) \\
	&\qquad\times \frac{(-iE/2+\sqrt{1-E^2/4})^3}{((1-r)(1-E^2/2-iE\sqrt{1-E^2/4})+1)^2} \Big\} + O(N^{-3/2})
\end{split} \]
for $|E| < 2-\delta$, and the implicit constant in the $O$-notation depends only on $\delta >0$.
\end{thm}

\vspace{2mm}\noindent
These results are based on a saddle-point analysis of the exact integral formula{\ae} for $\rho(E)$,
which we prove (in Lemmata~\ref{l:GUE}, \ref{l:GOE}, and  \ref{l:interp} below) using the supersymmetric formalism. The
supersymmetric formalism, put forth by Berezin (see \cite{B} for an early application to Wigner matrices)
and developed in the works of Wegner and Efetov, is a very general method to derive dual integral
representations for expressions such as an average product of several matrix elements of the resolvent.
While widely applied in the physical literature, only a fraction of these applications have been put
on rigorous mathematical basis.

On the other hand, the supersymmetric method is potentially applicable to a wide range of problems
pertaining to the spectral properties of random matrices and random operators; see the review of Spencer \cite{Sp}.

Two of the alternative groups of methods to study the eigenvalue distribution of random matrices
are perturbative methods (such as the moment method), and the method of orthogonal polynomials.
The moment method was applied by Wigner in the 1950's to prove the weak convergence of the
spectral distribution to the semicircle law $\wig$. A major disadvantage of all perturbative methods is
that they typically allow to control the density of states at some scale $\epsilon \sim N^{-\kappa}$,
i.e.\ they do not allow to take $\epsilon \to +0$ while keeping $N$ fixed (moreover, usually $\kappa < 1$,
so the perturbative methods are unable to see the oscillatory corrections to $\wig$). The
supersymmetric method allows to derive exact formul{\ae} for fixed $N$ and $\epsilon \to +0$.

The method of orthogonal polynomials, developed in the 1960's by Dyson, Gaudin, and Mehta (see
the book of Mehta \cite{M}), allows to compute the asymptotics of the density of states in the strong
sense and to arbitrary precision. For example, the asymptotic expansions of Theorems~\ref{thm:gue}
and \ref{thm:goe} (as well as analogous expansions for several other ensembles) were derived by
Forrester, Frankel, and Garoni \cite{FFG1,FFG2}. Theorem~\ref{thm:interp} can probably be extracted
via asymptotic analysis from formula (4.52) of Mehta and Pandey \cite{MP}; see also \S~5 there.

A vast generalization of Theorems~\ref{thm:gue}
and \ref{thm:goe} was obtained by Desroisiers and Forrester \cite{DF}, who considered general
$\beta$-ensembles (with arbitrary $\beta > 0$, where $\beta = 1$ corresponds to GOE and
$\beta = 2$ --- to GUE; the interpolating ensembles (\ref{eq:cov'}) are however not a special
case of $\beta$-ensembles). Their work is based on the study of multivariate Hermite polynomials.

On the other hand, the potential range of applicability of the supersymmetric method seems to include
many problems beyond the applicability of the orthogonal polynomial method (and even the method of multivariate
orthogonal polynomials); see \cite{Sp}.

\vspace{2mm}\noindent
Thus the supersymmetric method has several advantages over both perturbative methods and the
method of orthogonal polynomials. The applications found during the last dozen years include an
analysis of the density of states of a 3D band matrix model by Disertori, Pinson, and Spencer \cite{DPS},
and the study of mixed moments of characteristic polynomials for a class of 1D band matrices by
T.~Shcherbina \cite{Sh}; see \cite{Sp} for a review of other results.

Kalisch and Braak \cite{KB} used the supersymmetric formalism to derive a formula for GOE, GUE
and GSE density of states, and then applied saddle-point analysis to derive the asymptotics of
Theorems~\ref{thm:gue} and \ref{thm:goe}. Their work is however on the physical level of rigor.
A mathematically rigorous derivation of Theorem~\ref{thm:gue} (as well as of its counterparts at the
spectral edges) was given by Disertori \cite{D}. The derivation of the integral formul{\ae} by Kalisch--Braak
and Disertori is based on the Hubbard-Stratonovich transformation.

Our contribution is three-fold. First, we derive a integral representation for the density of states
using a different supersymmetric approach; our formal{\ae} (Lemmata \ref{l:GUE}, \ref{l:GOE}, 
and \ref{l:interp})
seem simpler than the ones obtained via the Hubbard--Stratonovich transformation. We mention
that a different approach avoiding the Hubbard--Stratonovich transformation was developed
by Fyodorov in \cite{F}.

Second, we perform a mathematically rigorous saddle-point analysis of both formal{\ae} to derive the
asymptotic expansions of Theorems~\ref{thm:gue}, \ref{thm:goe}, and \ref{thm:interp}, thus extending the work
of Disertori to GOE and to the interpolating ensembles. Although the results (at least, those
pertaining to GUE and GOE) are not new, we believe that the methods can be applied
to other problems intractable by other means; thus our third goal is a detailed and (relatively)
self-contained exposition.

\vspace{2mm}\noindent
We remark that the method of the current paper can be probably applied to other Gaussian
ensembles. As an example, we mention the anti-symmetric Hermitian ensemble of Mehta and Rosenzweig \cite{MR}, corresponding to $r = 2$ in (\ref{eq:cov'}).

\vspace{3mm}\noindent
To  state the integral formul{\ae}, we need some notation.
Let $E_\epsilon = E - i\epsilon$ and
\[ G(z) = (z - H)^{-1}~; \quad G(z; m,n) = (z-H)^{-1}(m,n)~. \]
The density of states can be expressed in terms of $G$ as follows:
\begin{equation}\label{eq:rhoG}
 \rho(E) = \Im \lim_{\epsilon \to +0} \frac{1}{\pi N} \tr \, \langle G(E_\epsilon) \rangle~.
\end{equation}
The supersymmetric formalism is used to derive a (dual) integral representation for $\langle \tr \, G(E_\epsilon) \rangle$.

\begin{lemma}\label{l:GUE} For GUE,
\begin{multline*}
\frac{1}{N}\langle \tr \, G(E_\epsilon) \rangle
  = \frac{(-1)^{N-1}N}{2\pi} \int_0^\infty ds \oint dz (iE - z + s)\\
  \exp \left[ - N (iE_\epsilon s + \frac{1}{2}s^2 - \ln s) \right]
  \exp \left[ - N (iE_\epsilon z - \frac{1}{2}z^2 + \ln z) \right] ~,
\end{multline*}
whereas
\begin{multline*}
i = \langle i \rangle
  = \frac{(-1)^{N-1}N}{2\pi} \int_0^\infty ds \oint dz \frac{iE - z + s}{s}\\
  \exp \left[ - N (iE_\epsilon s + \frac{1}{2}s^2 - \ln s) \right]
  \exp \left[ - N (iE_\epsilon z - \frac{1}{2}z^2 + \ln z) \right] ~.
\end{multline*}
\end{lemma}

\begin{lemma}\label{l:GOE} For GOE,
\begin{multline*}
\frac{1}{N}\langle \tr \, G(E_\epsilon) \rangle
  = \frac{(-1)^{N}2^NN^2}{8\pi^2} \int_0^\infty ds \int_0^\infty dt \int_{-1}^1 d\alpha \oint dz\\
  z \frac{s+t}{st}(1-\alpha^2)^{-3/2} \left[ \frac{1}{N} + (z-iE_\epsilon)^2 - 2 (z-iE_\epsilon)(s+t)
    + 4st(1-\alpha^2)\right]\\
  \exp \left[ - N (iE_\epsilon (s+t) + s^2 + t^2 - \frac{1}{2}\ln s - \frac{1}{2}\ln t) \right] \\
  \exp \left[ - N (iE_\epsilon z - \frac{1}{2}z^2 + \ln z) \right]
  \exp \left[ -N(2st \alpha^2 - \frac{1}{2} \ln(1-\alpha^2))\right]~,
\end{multline*}
whereas
\begin{multline*}
i = \langle i \rangle
= \frac{(-1)^{N}2^NN^2}{8\pi^2} \int_0^\infty ds \int_0^\infty dt \int_{-1}^1 d\alpha \oint dz\\
  z \frac{1}{st}(1-\alpha^2)^{-3/2} \left[ \frac{1}{N} + (z-iE_\epsilon)^2 - 2 (z-iE_\epsilon)(s+t)
    + 4st(1-\alpha^2)\right]\\
  \exp \left[ - N (iE_\epsilon (s+t) + s^2 + t^2 - \frac{1}{2}\ln s - \frac{1}{2}\ln t) \right] \\
  \exp \left[ - N (iE_\epsilon z - \frac{1}{2}z^2 + \ln z) \right]
  \exp \left[ -N(2st \alpha^2 - \frac{1}{2} \ln(1-\alpha^2))\right]~.
\end{multline*}
\end{lemma}

\begin{lemma}\label{l:interp} For the interpolating ensemble (\ref{eq:cov'}),
\begin{multline*}
\frac{1}{N}\langle \tr \, G(E_\epsilon) \rangle
  = \frac{(-1)^{N}2^NN^2}{8\pi^2} \int_0^\infty ds \int_0^\infty dt \int_{-1}^1 d\alpha \oint dz\\
  z \frac{s+t}{st}(1-\alpha^2)^{-3/2} \left[ \frac{1}{N} + (z-iE_\epsilon)^2 - (2-r) (z-iE_\epsilon)(s+t)
    + 4(1-r)st(1-\alpha^2)\right]\\
  \exp \left[ - N (iE_\epsilon (s+t) + (1 - \frac{r}{2}) (s^2 + t^2) + rst  - \frac{1}{2}\ln s - \frac{1}{2}\ln t) \right] \\
  \exp \left[ - N (iE_\epsilon z - \frac{1}{2}z^2 + \ln z) \right]
  \exp \left[ -N(2(1-r)st \alpha^2 - \frac{1}{2} \ln(1-\alpha^2))\right]~,
\end{multline*}
whereas
\begin{multline*}
i = \langle i \rangle 
  = \frac{(-1)^{N}2^NN^2}{8\pi^2} \int_0^\infty ds \int_0^\infty dt \int_{-1}^1 d\alpha \oint dz\\
  z \frac{1}{st}(1-\alpha^2)^{-3/2} \left[ \frac{1}{N} + (z-iE_\epsilon)^2 - (2-r) (z-iE_\epsilon)(s+t)
    + 4(1-r)st(1-\alpha^2)\right]\\
  \exp \left[ - N (iE_\epsilon (s+t) + (1 - \frac{r}{2}) (s^2 + t^2) + rst  - \frac{1}{2}\ln s - \frac{1}{2}\ln t) \right] \\
  \exp \left[ - N (iE_\epsilon z - \frac{1}{2}z^2 + \ln z) \right]
  \exp \left[ -N(2(1-r)st \alpha^2 - \frac{1}{2} \ln(1-\alpha^2))\right]~.
\end{multline*}
\end{lemma}

In the three lemmata, the contour integral is along a counterclockwise contour about zero; the choice of the
branch of the logarithm is not important, since it is multiplied by an integer number $N$ in the exponent.
We prove the lemmata in Section~\ref{s:susy}, and then apply saddle point analysis to derive the theorems in
Section~\ref{s:sad}. We omit the proofs of Lemma~\ref{l:interp} and Theorem~\ref{thm:interp}
which are almost identical to the proofs of Lemma~\ref{l:GOE} and Theorem~\ref{thm:goe}, respectively.

\vspace{3mm}\noindent
{\bf Acknowledgement} I am grateful to Tom Spencer for encouraging me to work on this problem, and
for his helpful advise. I also thank  Margherita Disertori, Tanya Shcher\-bina and Sasha Sodin for  useful discussions, and Yan Fyodorov for suggesting to consider the interpolating ensembles 
of Mehta--Pandey, and for his interest in this work.

\section{Integral representation}\label{s:susy}

The proof of Lemmata~\ref{l:GUE},\ref{l:GOE}, and \ref{l:interp} is based on the supersymmetric formalism. Let us introduce the notation.

For
$\z, \z' \in \mathbb{C}^N$, set
\[ [\z, \z'] = \sum_{j=1}^N \bar{z}_j z'_j~, \quad |\z| = \sqrt{[\z, \z]}~;\]
in particular, for an $N \times N$ matrix $A$,
\[ [\z, A\z] = \sum \bar{z}_j A_{jk} z_k~. \]
Also set
\begin{equation}\label{eq:dz}
D\z = \prod_{j=1}^N \frac{d\Re z_j \, d\Im z_j}{\pi}~.
\end{equation}
Let $\psi_1, \cdots, \psi_N, \bar{\psi}_1, \cdots, \bar{\psi}_N$ be anti-commuting variables, i.e.\
\[ \psi_i \psi_j + \psi_j \psi_i = \psi_i \bar\psi_j + \bar\psi_j \psi_i = \bar\psi_i \bar\psi_j + \bar\psi_j \bar\psi_i = 0~,\]
and let
\[ [\bpsi, \bpsi] = \sum \bar\psi_j \psi_j~, \quad [ \bpsi, A\bpsi] = \sum \bar\psi_j A_{jk} \psi_k~. \]
Also let
\[ D\bpsi = \prod_{j=1}^N d\bar\psi_j d\psi_j~.\]
The supersymmetric (Berezin) integral is defined by the rules
\[ \int \psi_j d\psi_j = \int \bar \psi_j d\bar\psi_j = 1~, \quad \int d\psi_j = \int d\bar\psi_j = 0~.\]

We start from the identities
\begin{equation}\label{eq:ident}
\begin{split}
(H - \E)^{-1}(k, j) &= i \int D\z \int D\bpsi  \, e^{-i[\z, (\E-H)\z] - i[\bpsi, (\E - H)\bpsi]} z_k \bar{z}_j~,\\
i                   &= i \int D\z \int D\bpsi  \, e^{-i[\z, (\E-H)\z] - i[\bpsi, (\E - H)\bpsi]}~,
\end{split}
\end{equation}
which are valid for any Hermitian matrix $H$ and any $\epsilon > 0$ (see Spencer \cite[(4.20)]{Sp}). We shall prove the
first part of each of the lemmata using the first identity; the second part is similarly derived from the second identity.

\begin{proof}[Proof of Lemma~\ref{l:GUE}]
Taking the expectation of (\ref{eq:ident}) and summing over $k = j$, we obtain:
\begin{equation}\label{eq:gue2} \frac{1}{N}\langle \tr \, G(E_\epsilon) \rangle
  = \frac{i}{N} \int D\z \int D\bpsi \, [\z, \z]\, e^{-i\E([\z,\z] + [\bpsi,\bpsi])}
  \left\langle e^{i ([\z, H\z] + [\bpsi, H\bpsi])} \right\rangle~.
\end{equation}
Using the  identity
\[ \langle e^g \rangle = e^{\langle g^2 \rangle/2} \]
(valid for any Gaussian random variable $g$), we deduce
\begin{equation}\label{eq:gue3}\frac{1}{N}\langle \tr \, G(E_\epsilon) \rangle
= \frac{i}{N} \int D\z \int D\bpsi \, [\z, \z]\, e^{-i\E([\z,\z] + [\bpsi,\bpsi]}
   e^{-\frac{1}{2} \langle ([\z, H\z] + [\bpsi, H\bpsi])^2 \rangle}~.
\end{equation}
From (\ref{eq:cov}),
\[ \langle ([\z, H\z] + [\bpsi, H\bpsi])^2 \rangle
  = \frac{1}{N} \left( [\z, \z]^2 - [\bpsi,\bpsi]^2 + 2 \sum_{j,k=1}^N {z}_j \bar\psi_j \bar{z}_k \psi_k\right)~.\]
Let $U$ be a unitary matrix such that $U_{1j} = \bar{z}_j / |\z|$. Denote
\[ \bar\phi_i = \sum_j \bar{U}_{ij} \bar{\psi}_j~, \quad \phi_i = \sum_j U_{ij} \psi_j~. \]
Then
\[ [\bphi, \bphi] = [\bpsi, \bpsi]~, \quad
\sum z_j \bar\psi_j = |\z|\bar\phi_1~, \quad \sum z_j \psi_j = |\z| \phi_1~, \quad D\bphi = D\bpsi~,\]
in particular,
\[ \sum_{j,k=1}^N z_j \bar\psi_j \bar{z}_k \psi_k = [\z,\z] \bar\phi_1\phi_1~. \]
Returning to (\ref{eq:gue3}), we obtain:
\begin{multline*} \frac{1}{N}\langle \tr \, G(E_\epsilon) \rangle
=  \frac{i}{N} \int D\z \int D\tphi d\bar\phi_1d\phi_1 \, [\z, \z] \, e^{-i\E([\z,\z] + [\tphi,\tphi])}
  e^{-i\E \bar\phi_1 \phi_1} \\
e^{-\frac{1}{2N} \left( [\z,\z]^2 - [\tphi, \tphi]^2 \right)}
e^{-\frac{1}{2N} \left( 2[\z,\z] \bar\phi_1 \phi_1 - 2 \bar\phi_1 \phi_1 [\tphi,\tphi]\right)}~,
\end{multline*}
where we have set $\tphi$ to be the Grassmanian vector $\bphi$ without the first coordinate. Integrating over $\phi_1$
(and $\bar\phi_1$), we obtain:
\[ \int d\bar\phi_1 d\phi_1 e^{-i\E \bar\phi_1 \phi_1 - \frac{1}{N} [\z,\z] \bar\phi_1 \phi_1
  - \frac{1}{N} \bar\phi_1 \phi_1 [\tphi,\tphi]} = i\E + \frac{1}{N} [\z,\z] - \frac{1}{N} [\tphi,\tphi]~.
  \]
To integrate over $\tphi$, we use a Grassmanian version of polar coordinates:
\begin{equation}\label{eq:phipolar}
\int D\tphi F([\tphi, \tphi]) = (-1)^{N-1} \frac{(N-1)!}{2\pi i} \oint \frac{F(z)}{z^N} dz
\end{equation}
(the contour encircles the origin counterclockwise).  For the sake of completeness, let us prove this formula, after stating it as
\begin{lemma}
For any analytic function $F$ and an $N$-component Grassmann vector $\bphi$,
\[\int D\bphi F([\bphi, \bphi]) = (-1)^{N} \frac{N!}{2\pi i} \oint \frac{F(z)}{z^{N+1}} dz~, \]
where the contour encircles the origin counterclockwise.
\end{lemma}
\begin{proof}
First,
\[ F([\bphi, \bphi])  = \sum_{j=0}^N \frac{F^{(j)}(0)}{j!} [\bphi, \bphi]^j~.  \]
Only the last term contributes to the integral, thus
\[ \int D\bphi \, F([\bphi, \bphi])
=  \frac{F^{(N)}(0)}{N!} \int D\bphi  [\bphi, \bphi]^N = (-1)^N F^{(N)}(0)~, \]
where on the second step we opened all the brackets and applied the integration rules. Now the statement
follows from Cauchy's formula.
\end{proof}
Applying (\ref{eq:phipolar}), we obtain:
\begin{multline}\label{eq:gue6}
\frac{1}{N}\langle \tr \, G(E_\epsilon) \rangle
= \frac{i}{N} (-1)^{N-1} \frac{(N-1)!}{2\pi i} \int D\z \oint \frac{dz}{z^N} \,[\z, \z]\,\\
\left( i\E + \frac{1}{N} [\z,\z] - \frac{z}{N} \right)
\exp\left\{-i\E [\z,\z] - \frac{1}{2N} [\z,\z]^2 - i\E z + \frac{1}{2N} z^2\right\}~.
\end{multline}
Now we pass to polar coordinates in $\z$ using the formula
\begin{equation}\label{eq:zpolar} \int_{\mathbb{C}^N}  D\z \, F([\z,\z]) = \frac{1}{(N-1)!} \int_0^\infty F(s) s^{N-1} ds~,
\end{equation}
and obtain:
\begin{multline}\label{eq:gue7}
\frac{1}{N}\langle \tr \, G(E_\epsilon) \rangle
= \frac{(-1)^{N-1}}{2\pi N} \int_0^\infty ds \oint dz (i\E + \frac{s-z}{N}) \\
\exp\left\{-i\E s - \frac{1}{2N} s^2\right\} \exp\left\{-i\E z + \frac{1}{2N} z^2\right\} \frac{s^N}{z^N}~.
\end{multline}
The change of variables $s \rightarrow Ns$, $z \rightarrow Nz$ concludes the proof of Lemma~\ref{l:GUE}.
\end{proof}

\begin{proof}[Proof of Lemma~\ref{l:GOE}]
Similarly to the proof of Lemma~\ref{l:GUE},
\begin{equation}\label{eq:goe3}\frac{1}{N}\langle \tr \, G(E_\epsilon) \rangle
= \frac{i}{N} \int D\z \int D\bpsi \, [\z, \z]\, e^{-i\E([\z,\z] + [\bpsi,\bpsi])}
   e^{-\frac{1}{2} \langle ([\z, H\z] + [\bpsi, H\bpsi])^2 \rangle}~.
\end{equation}
Now we have:
\[\begin{split} &\langle ([\z, H\z] + [\bpsi, H\bpsi])^2 \rangle\\
   &\quad= \frac{1}{N} \left\{ |\sum_j z_j^2|^2 + [\z, \z]^2 - [\bpsi,\bpsi]^2 + 2 \sum_j \bar{z}_j \bar\psi_j \sum_k z_k \psi_k
 + 2 \sum_j {z}_j \bar\psi_j \sum_k \bar{z}_k \psi_k\right\} \\
   &\quad= \frac{1}{N} \left\{ 2[\x, \x]^2 + 2[\y,\y]^2 + 4 [\x,\y]^2 + 4 \left[ \sum_j x_j \bar\psi_j \sum_k x_k \psi_k
+ \sum_j y_j \bar{\psi}_j \sum_k y_k \psi_k\right] \right\}~,
\end{split}\]
where we use the  decomposition $\z = \x + i\y$ of $\z$ into its real and imaginary parts.

Denote $u = \frac{1}{2} (1 + \frac{[\x,\y]}{|\x||\y|})$; then
\[ u(1-u) = \frac{1}{|\x|^2|\y|^2}(|\x|^2|\y|^2 - [\x,\y]^2)~. \]
Define
\[ \xx = |x| (\sqrt{u} \e_1 + \sqrt{1-u} \e_2)~, \quad \yy = |\y| (\sqrt{u} \e_1 - \sqrt{1-u}\e_2)~,
\]
where $\e_j$ is the $j$-th vector of the standard basis. Then one can find an orthogonal map
which takes $\x,\y$ to $\xx, \yy$. Hence, similarly to the proof of Lemma~\ref{l:GUE},  one can pass from
$\bpsi$ to a new Grassmann variable $\bphi$ and rewrite (\ref{eq:goe3}) as
\begin{multline*}
\frac{1}{N}\langle \tr \, G(E_\epsilon) \rangle
= \frac{i}{N} \iint D\x D\y  \int D\bphi (|\x|^2 + |\y|^2) \\
\exp\Big\{-i\E(|\x|^2 + |\y|^2 + [\bphi,\bphi])
-\frac{1}{2N} \Big[ 2 |\x|^4 + 2 |\y|^4 + 4 [\x,\y]^2 - [\bphi,\bphi]^2 \\+ 4 (|\x|^2+|\y|^2)(u \bar\phi_1\phi_1 + (1-u)\bar\phi_2\phi_2) +
4 \sqrt{u(1-u)}(|\x|^2-|\y|^2)(\bar\phi_2\phi_1 + \bar\phi_1\phi_2)\Big]\Big\}~,
\end{multline*}
where $D\x D\y$ still incorporates the factor $\pi^{-N}$ from (\ref{eq:dz}).
Let $\tphi$ be the Grassmannian vector $\bphi$ whithout the first two coordinates. Then the above expression can be rewritten as
\begin{multline*}
\frac{1}{N}\langle \tr \, G(E_\epsilon) \rangle
= \frac{i}{N} \iint D\x D\y  \int D\tphi \int d\bar\phi_1d\phi_1d\bar\phi_2d\phi_2 (|\x|^2 + |\y|^2) \\
\exp\Big\{-i\E(|\x|^2 + |\y|^2 + [\tphi,\tphi])
-\frac{1}{N} \Big[ |\x|^4 + |\y|^4 + 2 [\x,\y]^2 - \frac{1}{2} [\tphi,\tphi]^2 \Big]\\
- i\E (\bar\phi_1\phi_1 + \bar\phi_2\phi_2) + \frac{1}{N}\Big[\bar\phi_1\phi_1\bar\phi_2\phi_2
    + [\tphi,\tphi](\bar\phi_1\phi_1 + \bar\phi_2\phi_2) \\
- 2 (|\x|^2+|\y|^2)(u \bar\phi_1\phi_1 + (1-u)\bar\phi_2\phi_2)
  -2  \sqrt{u(1-u)}(|\x|^2-|\y|^2)(\bar\phi_2\phi_1 + \bar\phi_1\phi_2)\Big]\Big\}~.
\end{multline*}

Let us first integrate over $\phi_1$, $\bar\phi_1$, $\phi_2$, and $\bar\phi_2$. Using the Taylor expansion
$e^z = 1 + z + z^2/2+ \cdots$ (the higher-order terms vanish due to anticommutativity), we obtain:
\begin{equation}
\begin{split}
&\int d\bar\phi_1d\phi_1d\bar\phi_2d\phi_2 \exp\Big\{ \frac{1}{N} \bar\phi_1\phi_1\bar\phi_2\phi_2
 + (\frac{1}{N} [\tphi,\tphi]-i\E) (\bar\phi_1\phi_1 + \bar\phi_2\phi_2) \\
&\qquad\qquad - \frac{2}{N} (|\x|^2+|\y|^2)(u \bar\phi_1\phi_1 + (1-u)\bar\phi_2\phi_2)  \\
&\qquad\qquad- \frac{2}{N}  \sqrt{u(1-u)}(|\x|^2-|\y|^2)(\bar\phi_2\phi_1 + \bar\phi_1\phi_2)\Big\} \\
&\qquad= \frac{1}{N} + (\frac{1}{N} [\tphi,\tphi]-i\E)^2 - \frac{2}{N}(\frac{1}{N} [\tphi,\tphi]-i\E)(|\x|^2+|\y|^2)\\
&\qquad\qquad+ \frac{4}{N^2}(|\x|^2|\y|^2 - [\x,\y]^2)~.
\end{split}
\end{equation}
Thus
\begin{equation}
 \begin{split}
&\frac{1}{N}\langle \tr \, G(E_\epsilon) \rangle
= \frac{i}{N} \iint D\x D\y  \int D\tphi\,\, (|\x|^2 + |\y|^2) \\
&\qquad\qquad \Big( \frac{1}{N} + (\frac{1}{N} [\tphi,\tphi]-i\E)^2 - \frac{2}{N}(\frac{1}{N} [\tphi,\tphi]-i\E)(|\x|^2+|\y|^2)\\
&\qquad\qquad\qquad+ \frac{4}{N^2}(|\x|^2|\y|^2 - [\x,\y]^2) \Big)\\
&\qquad\qquad\exp\Big\{-i\E([\x,\x] + [\y,\y] + [\tphi,\tphi]) \\
&\qquad\qquad\qquad-\frac{1}{N} \Big[ |\x|^4 + |\y|^4 + 2 [\x,\y]^2 - \frac{1}{2} [\tphi,\tphi]^2 \Big]\Big\}~.
 \end{split}
\end{equation}
Now we integrate over $\tphi$ using the formula (\ref{eq:phipolar}). We obtain:
\begin{equation}
 \begin{split}
&\frac{1}{N} \langle \tr \, G(E_\epsilon) \rangle
= \frac{i}{N} \frac{(-1)^N (N-2)!}{2\pi i} \iint D\x D\y  \oint \frac{dz}{z^{N-1}} \, (|\x|^2 + |\y|^2) \\
&\qquad\qquad \Big( \frac{1}{N} + (\frac{z}{N}-i\E)^2 - \frac{2}{N}(\frac{z}{N}-i\E)(|\x|^2+|\y|^2)\\
&\qquad\qquad\qquad+ \frac{4}{N^2}(|\x|^2|\y|^2 - [\x,\y]^2) \Big)\\
&\qquad\qquad\exp\Big\{-i\E([\x,\x] + [\y,\y] + z) \\
&\qquad\qquad\qquad-\frac{1}{N} \Big[ |\x|^4 + |\y|^4 + 2 [\x,\y]^2 - \frac{1}{2} z^2 \Big]\Big\}~.
 \end{split}
\end{equation}
Finally, we pass to polar coordinates in $\x$ and $\y$. Setting $x = s\alpha$, $y = t\beta$, where $s=|\x|$,
$y=|\y|$, and $\alpha,\beta \in S^{N-1}$ and using the formul{\ae}
\[
\int_{\mathbb{R}^N} F(\x) D\x = \frac{N\pi^{N/2}}{\Gamma(\frac{N}{2} +1)} \int_0^\infty ds \int_{S^{N-1}} d\sigma(\alpha)
  F(s\alpha)
\]
and
\[
\iint_{S^{N-1} \times S^{N-1}} F([\alpha,\beta]) d\sigma(\alpha) d\sigma(\beta)
= \frac{1}{B(\frac{N-1}{2},\frac{1}{2})} \int_{-1}^1 F(\alpha_1) (1-\alpha_1)^{\frac{N-3}{2}} d\alpha_1~,
\]
where $\sigma$ is the invariant probability measure on $S^{N-1}$, and
\[ B(a,b) = \int_0^1 t^{a-1} (1-t)^{b-1} dt = \frac{\Gamma(a)\Gamma(b)}{\Gamma(a+b)}\]
is the Euler beta function, we obtain:
\begin{equation}
 \begin{split}
&\frac{1}{N} \langle \tr \, G(E_\epsilon) \rangle
= \frac{i}{N \pi^N} \frac{(-1)^N (N-2)!}{2\pi i} \frac{\pi^N N^2}{\Gamma(N/2+1)^2 B(\frac{N-1}{2},\frac{1}{2})} \\
&\qquad \int_0^\infty ds \, s^{N-1} \int_0^\infty dt \, t^{N-1} \int_{-1}^1 d\alpha \, (1-\alpha^2)^{\frac{N-3}{2}}\oint \frac{dz}{z^{N-1}} \\
&\qquad\qquad (s^2 + t^2) \Big( \frac{1}{N} + (\frac{z}{N}-i\E)^2 - \frac{2}{N}(\frac{z}{N}-i\E)(s^2 + t^2)\\
&\qquad\qquad\qquad+ \frac{4}{N^2}s^2 t^2 (1 - \alpha^2) \Big)\\
&\qquad\qquad\exp\Big\{-i\E(s^2 + t^2  + z) -\frac{1}{N} \Big[ s^4 + t^4 + 2 s^2t^2\alpha^2 - \frac{1}{2} z^2 \Big]\Big\}~.
 \end{split}
\end{equation}
The final change of variables $s \gets s^2/N$, $t \gets t^2/N$, $z \gets z/N$ concludes the proof.
\end{proof}

\section{Saddle point analysis}\label{s:sad}

\begin{proof}[Proof of Theorem~\ref{thm:gue}]
We shall take $\epsilon \to +0$ in Lemma~\ref{l:GUE} and compute the asymptotics using saddle-point analysis.
First we calculate the saddles. Set
\[ f(s) = iEs + \frac{1}{2} s^2 - \ln s~.\]
Then
\[ f'(s) = iE + s - \frac{1}{s}~,\]
therefore the saddles are
\begin{equation}\label{eq:spm} s_\pm = - \frac{iE}{2} \pm \sqrt{1-\frac{E^2}{4}}~.\end{equation}
Similarly, for
\[ g(z) = iE z - \frac{1}{2} z^2 + \ln z \]
we have
\[ g'(z) = iE - z + \frac{1}{z}~, \]
so the saddles are
\begin{equation}\label{eq:zpm} z_\pm = \frac{iE}{2} \pm \sqrt{1-\frac{E^2}{4}}~. \end{equation}
We deform the contours in $s$ and $z$ as follows:
\begin{equation}\label{eq:gamma1} \Gamma_1: \quad s = \begin{cases}
s_+ \tilde{s}~, &0 \leq \tilde{s} \leq A \\
s_+ A + \tilde{s} - A~, &A \leq \tilde{s}~,
       \end{cases}\end{equation}
and
\begin{equation}\label{eq:gamma2} \Gamma_2 : \quad z = e^{i\theta}~, \quad 0 \leq \theta \leq 2\pi~, \end{equation}
where
\[ A = \begin{cases}
2~, &|E| \leq \sqrt{3}\\
\frac{1}{\frac{E^2}{2}-1}~, &\sqrt{3} < |E| < 2
       \end{cases}~. \]
The contour in $s$ passes through the saddle point $s_+$, whereas the contour in $z$ passes
through both $z_+$ and $z_-$. The change of coutour is justified according to Cauchy's theorem.
\begin{cl}\label{cl-f}
The minimum of $\Re f(s)$ on the $s$-contour is achieved at $s = s_+$, i.e.\ $\tilde{s}=1$.
\end{cl}
\begin{proof}[Proof of Claim~\ref{cl-f}]
We have:
\[ \Re f(x+iy) = - E y + \frac{1}{2}(x^2 - y^2) - \frac{1}{2}\ln(x^2 + y^2)~.\]
Therefore
\[\begin{split}
\frac{\partial}{\partial x} \Re f(x+iy) &= x \left\{ 1 - \frac{1}{x^2 + y^2} \right\}~,\\
\frac{\partial}{\partial y} \Re f(x+iy) &= - E - y \left\{ 1 + \frac{1}{x^2 + y^2} \right\}~.
\end{split}\]
For $x^2 + y^2 \geq 1$, the $x$-derivative is positive, therefore $\Re f(s(\tilde{s}))$
is increasing for $\tilde{s} \geq A$. Next,
\[\begin{split} \frac{d}{d\tilde{s}} \Re f(s_+ \tilde{s})
&= \tilde{s} (\Re s_+)^2 (1 - \tilde{s}^{-2}) - \Im s_+ (E +  \tilde{s} \Im s_+ (1 + \tilde{s}^{-2}))\\
&= \tilde{s} (1 - E^2/4) (1 - \tilde{s}^{-2}) + \frac{E}{2} (E - \tilde{s} \frac{E}{2} (1 + \tilde{s}^{-2})) \\
&= \tilde{s}^{-1} \left\{ \tilde{s}^2 (1 - \frac{E^2}{2}) + \tilde{s} \frac{E^2}{2} - 1 \right\}~.
\end{split}\]
The quadtatic  expression in the brackets has two roots, $1$ and $(E^2/2-1)^{-1}$; for $|E| < \sqrt{2}$ the second root
is negative, whereas for $|E| > \sqrt{2}$ it  is greater than $A$. Therefore $\Re f(s(\tilde{s}))$ is decreasing for
$0 \leq \tilde{s} \leq 1$ and increasing for $1 \leq \tilde{s} \leq A$.
\end{proof}

\begin{cl}\label{cl-g}
The minimum of $\Re g(z)$ on the $z$-contour is achieved at $z = z_\pm$, i.e.\ for the two values of $\theta$ for
which $\sin \theta = \frac{E}{2}$.
\end{cl}
\begin{proof}[Proof of Claim~\ref{cl-g}]
We have:
\[ \Re f(e^{i\theta}) = - E \sin \theta - \frac{1}{2} (\cos^2 \theta - \sin^2 \theta)~,\]
hence
\[ \frac{d}{d\theta} \Re f(e^{i\theta}) = - \cos \theta (E - 2 \sin \theta)~. \]
The claim easily follows.
\end{proof}
According to Claims~\ref{cl-f} and \ref{cl-g}, the saddle-point approximation is justified, i.e.\ the
asympototics of the integral
\[ I = \frac{(-1)^{N-1}N}{2\pi} \int_{\Gamma_1} ds \oint_{\Gamma_2} dz (iE - z + s)
\exp \left\{ - N(f(s) + g(z)) \right\}\]
is given (to arbitrary order in $1/N$) by the contribution of the saddle points $(s= s_+, z=z_+)$ and
$(s=s_+, z = z_-)$. This follows from the general results on saddle-point approximation, see e.g.\
Fedoryuk \cite[\S 4.1, Theorem~1.2]{F}.

Let us compute the contribution of the saddle points (up to order $1/N$).

Denote $C_N = \frac{(-1)^{N-1}N}{2\pi}$. Then the second part of Lemma~\ref{l:GUE} yields:
\[ i = \langle i \rangle = C_N
\int_{\Gamma_1} ds \oint_{\Gamma_2} dz \frac{iE - z + s}{s}
\exp \left\{ - N(f(s) + g(z)) \right\}~,\]
therefore
\[ I = i s_+ + \frac{(-1)^{N-1}N}{2\pi} \int_{\Gamma_1} ds \oint_{\Gamma_2} dz \frac{iE - z + s}{s} (s-s_+)
\exp \left\{ - N(f(s) + g(z)) \right\}~. \]
We have:
\[ f''(s) = 1 + s^{-2}~, \quad f'''(s) = - 2 s^{-3}~,\]
hence
\[ f(s_+) = iE s_+ + \frac{1}{2} s_+^2 - \ln s_+~, \quad
f''(s_+) = \frac{s_+^2+1}{s_+^2}~, \quad f'''(s_+) = - \frac{2}{s_+^3}~.
\]
Also, for
\[ \phi(s) = \frac{iE - z + s}{s} (s-s_+)~, \]
we have:
\[ \phi'(s) = 1 + (iE - z) \frac{s_+}{s^2}~, \quad \phi''(s) = - (iE-z) \frac{2s_+}{s^3}~, \]
hence
\[ \phi(s_+) = 0~, \quad \phi'(s_+) = \frac{iE + s_+ - z}{s_+}~, \quad \phi_2(s_+) = - 2 \frac{i E - z}{s_+}~. \]
Therefore
\[\begin{split}
&\sqrt{\frac{2\pi}{N}} \frac{e^{-Nf(s_+)}}{N (f''(s_+))^{3/2}} \left[ \frac{\phi''(s_+)}{2} - \frac{f'''(s_+)\phi'(s_+)}{2f''(s_+)}\right]\\
&\quad= \sqrt{\frac{2\pi}{N}} \frac{1}{N} e^{-N(iEs_+ + \frac{1}{2}s_+^2 - \ln s_+)} \frac{s_+}{(s_+^2+1)^{3/2}} \left[ z - iE +
  \frac{iE + s_+ - z}{s_+^2-1}\right]~.
  \end{split}\]
This expression gives the contribution of $s_+$ to the $s$-integral, up to corrections of order $N^{-3/2}$. The contribution
of $z_\pm$ to the $z$-integral is given by
\[ \sqrt{\frac{2\pi}{N}} \left[ \frac{e^{-Ng(z_+)}}{\sqrt{g''(z_+)}} \psi(z_+)+ \frac{e^{-Ng(z_-)}}{\sqrt{g''(z_-)}} \psi(z_-)\right]~,\]
where $\psi(z)$ is the prefactor in the $z$-integral. This expression is equal to
\[ e^{-N(iE z_+ - \frac{1}{2}z_+^2 + \ln z_+)} \frac{iz_+s_+}{(z_+^2 + 1)^{1/2}}~.\]
Note that the second term vanishes since $\psi(z_-) = 0$.
Combining all the expressions, we obtain:
\[\begin{split}
I &= \frac{E}{2} + i \sqrt{1-\frac{E^2}{4}} \\
&\quad- \frac{(-1)^{N}i}{4N(1-E^2/4)} e^{-N \left[iE \sqrt{1-E^2/4} + 2 \ln (iE/2 + \sqrt{1-E^2/4})\right]}
+ O(N^{-3/2})~.
\end{split}\]
We conclude the proof by taking the imaginary part and using (\ref{eq:rhoG}).
\end{proof}

\begin{proof}[Proof of Theorem~\ref{thm:goe}]
As in the proof of Theorem~\ref{thm:goe}, we take the limit $\epsilon \to +0$ in the formula from
Lemma~\ref{l:GOE}, and use saddle-point analysis to compute the asymptotics of the resulting integral
\[
I = \frac{(-1)^{N}2^NN^2}{8\pi^2} \int_0^\infty ds \int_0^\infty dt \int_{-1}^1 d\alpha \oint dz
\exp\left\{-N F(s,t,z,\alpha)\right\} \Phi(s,t,z,\alpha)~,
\]
where
\begin{multline*}F(s,t,z,\alpha) = iE(s+t) + s^2 + t^2 - \frac{1}{2}\ln s - \frac{1}{2} \ln t \\
+ iEz - \frac{1}{2} z^2 + \ln z
+2st\alpha^2 - \frac{1}{2}\ln(1-\alpha^2)
\end{multline*}
and
\begin{multline*}
\Phi(s,t,z,\alpha) =  z \frac{s+t}{st}(1-\alpha^2)^{-3/2} \\
\times \left[ \frac{1}{N} + (z-iE_\epsilon)^2 - 2 (z-iE_\epsilon)(s+t)
    + 4st(1-\alpha^2)\right]~.
\end{multline*}
The relevant saddles of $F$ are given by $\alpha = 0$, $s = t = \frac{s_+}{2}$, $z = z_\pm$, where $s_+$ is the
same as in (\ref{eq:spm}) and $z_\pm$ is as in (\ref{eq:zpm}). We shall deform the contours so that they will pass
through these saddles and the minimum of $\Re F$ will be achieved only at these two points. We do it as follows: the $\alpha$-contour will
remain the interval $[-1, 1]$. In the $s$ and $t$-variables, we integrate along the contour $\Gamma_1$ from (\ref{eq:gamma1}),
with the modification
\[ A = \begin{cases}
2~, &|E| \leq \sqrt{\frac52}\\
\frac{1}{\sqrt{E^2/2-1}}~, &\sqrt{\frac52} < |E| < 2\\
       \end{cases}~,\]
while in the $z$-variable, we integrate along $\Gamma_2$ from (\ref{eq:gamma2}).
\begin{cl}\label{cl-al}
For every $s,t \in \Gamma_1$ and $z \in \Gamma_2$, the minimum of $\Re F(s,t,z,\alpha)$ on $[-1, 1]$ is achieved at
the point $\alpha = 0$.
\end{cl}
\begin{proof}[Proof of Claim~\ref{cl-al}]
We have:
\[ \frac{\partial \Re F}{\partial \alpha} = 4 \Re st \alpha + \frac{\alpha}{1-\alpha^2} =
\frac{\alpha}{1-\alpha^2} ((4 \Re st + 1) - 4 \Re st \alpha^2)~. \]
The derivative vanishes at $\alpha = 0$, and at the two points $\alpha_\pm$ given by
\[ \alpha^2 = 1 + \frac{1}{4\Re st}~.\]
For $s,t$ on the contour $\Gamma_1$,
\[ \Re st \geq A^2 \Re s_+^2 \geq 0~,\]
hence the two points are not in the domain $[-1, 1]$.
Since $\Re f$ tends to $+\infty$ as $\alpha \to \pm 1$, the minimum is
indeed at $\alpha = 0$.
\end{proof}
According to Claim~\ref{cl-al} and Claims~\ref{cl-f} (modified for the new choice of $A$) and \ref{cl-g}, the
minimum of $F$ is indeed achieved at the two saddles. Thus the asympototic contribution to $I$ comes, to any order
in $N^{-1}$, ony from the neighborhoods of these saddles.

Since
\[  \frac{(-1)^{N}2^NN^2}{8\pi^2} \int_0^\infty ds \int_0^\infty dt \int_{-1}^1 d\alpha \oint dz
\exp\left\{-N F(s,t,z,\alpha)\right\} \frac{\Phi(s,t,z,\alpha)}{s+t} = i\]\
according to the second half of Lemma~\ref{l:GOE}, we rewrite $I$ as
\[ I = i s_+ + C_N \int_0^\infty ds \int_0^\infty dt \int_{-1}^1 d\alpha \oint dz
e^{ -N(f(s)+f(t)+g(s,t,\alpha)+h(z))} \Phi_1(s,t,z,\alpha)~,
\]
where
\[\begin{split}
C_N &= \frac{(-1)^{N}2^NN^2}{8\pi^2}~, \\
f(s) &=  iEs + s^2 -  \frac{1}{2}\ln s~, \\
g(s,t,\alpha) &= 2st\alpha^2 - \frac{1}{2}\ln(1-\alpha^2)~,\\
h(z) &= iEz - \frac{1}{2} z^2 + \ln z~, \\
\Phi_1(s,t,z,\alpha) &= z \frac{s- \frac{s_+}{2}+t - \frac{s_+}{2}}{st}(1-\alpha^2)^{-3/2} \\
&\quad \times \left[ \frac{1}{N} + (z-iE_\epsilon)^2 - 2 (z-iE_\epsilon)(s+t)
    + 4st(1-\alpha^2)\right]~.
  \end{split}\]
The addend $1/N$ does not contribute to the asymptotics up to order $1/N$. Therefore we  compute the leading term of
the contribution of $\alpha = 0$ to the integral
\begin{equation}
\begin{split}
&\int e^{-N(2st\alpha^2 - \frac{1}{2} \ln(1-\alpha^2))} (1-\alpha^2)^{-3/2} \\
&\qquad\left[(z - iE)^2 - 2 (z-iE)(s+t) + 4st(1-\alpha^2)\right] d\alpha.
\end{split}
\end{equation}
It is equal to
\[
\sqrt{\frac{2\pi}{N}} \frac{1}{\sqrt{4st+1}} \left[(z - iE)^2 - 2 (z-iE)(s+t) + 4st\right]~, \]
up to terms of higher order. Now we integrate over $s$ and $t$, keeping the terms up to order $1/N$.
The integral is given by
\begin{equation}
I_2(z) = \int_0^\infty ds \int_0^\infty dt \, e^{-N  (f(s)+f(t))} \Phi_2(s,t,z)~,
\end{equation}
where
\[
\Phi_2(s, t, z) = \frac{s- \frac{s_+}{2}+t - \frac{s_+}{2}}{st \sqrt{4st+1}}
\left[(z-iE_\epsilon)^2 - 2 (z-iE_\epsilon)(s+t)
    + 4st\right]~;
\]
we compute the contribution of $s  = t = s_+$ to order $1/N$, which is equal to
\[ \frac{2\pi}{N^2} e^{-2Nf(s_+/2)} \frac{1}{f''(s_+/2)^2} \left[ \Phi_{2\,ss}(s_+/2,s_+/2,z) - \frac{f'''(s_+/2)}{f''(s_+/2)} \Phi_{2\, s}(s_+/2,s_+/2,z) \right] \]
since $\Phi_2(s_+/2,s_+/2,z) = 0$. This expression is equal to
\begin{multline*} \frac{2\pi}{N^2} e^{-N(iEs_+ + s_+^2/2 - \ln \frac{s_+}{2})} \frac{2s_+^2}{(s_+^2+1)^{5/2}}\\
 \times \left[ 2s_+ - 2(z-iE) - \frac{3s_+}{s_+^2+1} \left( (z-iE)^2 - 2 s_+ (z-iE) + s_+^2 \right)\right]~.
\end{multline*}
Then we compute the contribution of $z = z_\pm$ to the integral over $z$,
\[ \oint dz e^{-N h(z)} \Phi_3(z)~,\]
where
\[ \Phi_3(z) = z  \left[ 2s_+ - 2(z-iE) - \frac{3s_+}{s_+^2+1} \left( (z-iE)^2 - 2 s_+ (z-iE) + s_+^2 \right)\right]~. \]
The contribution comes from $z_-$, since $\Phi_3(z_+)=0$, and we only need the leading term,
\[ \sqrt{\frac{2\pi}{N}} \frac{1}{\sqrt{h''(z_-)}} e^{-Nh(z_-)} \Phi_3(z_-)~. \]
This final computation yields the answer:
\[ I = i s_+ + i \frac{\frac{iE}{2} - \sqrt{1-\frac{E^2}{4}}}{4N(1-\frac{E^2}{4})} + O(N^{-3/2})~. \]
In particular,
\[ \rho(E) = \Im \pi^{-1} I = \wig(E) - \frac{1}{4\pi^2 N \wig(E)} + O(N^{-3/2})~.\]
\end{proof}


\begin{thebibliography}{99}

\bibitem{B} F.~A.~Berezin,
Some remarks on the Wigner distribution,
Teoret.\ Mat.\ Fiz.~17 (1973), 305--318.

\bibitem{DF} P.~Desrosiers, P.~J.~Forrester,
Hermite and Laguerre $\beta$-ensembles: asymptotic corrections to the eigenvalue density,
Nuclear Phys.\ B 743 (2006), no.~3, 307--332.

\bibitem{D}  M.~Disertori,
Density of states for GUE through supersymmetric approach,
Rev.\ Math.\ Phys.~16 (2004), no.~9, 1191--1225.

\bibitem{DPS} M.~Disertori, H.~Pinson, T.~Spencer,
Density of states for random band matrices,
Comm.\ Math.\ Phys.~232 (2002), no.~1, 83--124.

\bibitem{F} M.~V.~Fedoryuk,
The saddle-point method,
``Nauka'', Moscow, 1977.  368 pp.

\bibitem{FFG2} P.~J.~Forrester, N.~E.~Frankel, T.~M.~Garoni,
Asymptotic form of the density profile for Gaussian and Laguerre random matrix ensembles
  with orthogonal and symplectic symmetry,
J.\ Math.\ Phys.~47 (2006), no.~2, 023301.

\bibitem{F} Y.~V.~Fyodorov,
Complexity of Random Energy Landscapes, Glass Transition, and Absolute Value of the Spectral 				Determinant of Random Matrices,
Phys.\ Rev.\ Lett.~92, no.~24, 240601 (2004) and 93, 149901 (2004).

\bibitem{FFG1} T.~M.~Garoni, P.~J.~Forrester, N.~E.~Frankel,
Asymptotic corrections to the eigenvalue density of the GUE and LUE,
J.\ Math.\ Phys.~46 (2005), no.~10, 103301.

\bibitem{KB} F.~Kalisch, D.~Braak,
Exact density of states for finite Gaussian random matrix ensembles via supersymmetry,
J.\ Phys.~A 35 (2002), no.~47, 9957--9969.

\bibitem{M}  M.~L.~Mehta,
Random matrices,
Pure and Applied Mathematics (Amsterdam), 142,
Elsevier/Academic Press, Amsterdam, 2004, xviii+688 pp, ISBN: 0-12-088409-7.

\bibitem{MR} M.~L.~Mehta, N.~Rosenzweig,
Distribution laws for the roots of a random antisymmetric hermitian matrix,
Nucl.\ Phys.~A 109.2 (1968): 449--456.

\bibitem{MP} A.~Pandey, M.~L.~Mehta,
Gaussian ensembles of random Hermitian matrices intermediate between 
orthogonal and unitary ones,
Comm.\ Math. Phys.~87.4 (1983): 449--468.

\bibitem{Sh} T.~Shcherbina,
On the second mixed moment of  characteristic polynomials of  1D band matrices,
arXiv:1209.3385

\bibitem{Sp} T.~Spencer,
SUSY statistical mechanics and random band matrices,
Quantum many body systems, 125--177,
Lecture Notes in Math., 2051, Springer, Heidelberg, 2012.

\end{thebibliography}
\end{document}